\documentclass[12pt,leqno]{article}
\usepackage{amsmath,amssymb,amsthm}
\newcommand{\defeq}{\; {\overset{\text{def}}=} \;}
\DeclareMathOperator{\Ad}{Ad}
\DeclareMathOperator{\ad}{ad}

\newcommand{\calB}{{\mathcal B}}
\newcommand{\calF}{{\mathcal F}}
\newcommand{\calL}{{\mathcal L}}
\newcommand{\calM}{{\mathcal M}}
\newcommand{\calP}{{\mathcal P}}
\newcommand{\calQ}{{\mathcal Q}}
\newcommand{\calX}{{\mathcal X}}
\newcommand{\commentout}[1]{}
\newcommand{\der}{\partial}

\DeclareMathOperator{\ord}{ord}
\DeclareMathOperator{\ordh}{ord^{\hbar}}
\DeclareMathOperator{\ordx}{ord_\xi}
\DeclareMathOperator{\Res}{Res}
\newcommand{\symh}{\sigma^\hbar}
\newcommand{\totsym}{\sigma_{\mathrm{tot}}}
\newtheorem{thm}{Theorem}[section]
\newtheorem{lem}[thm]{Lemma}
\newtheorem{prop}[thm]{Proposition}

\theoremstyle{definition}

\theoremstyle{remark}

\newtheorem{rem}[thm]{Remark}

\numberwithin{equation}{section}

\newcommand\thmref[1]{Theorem~\ref{#1}}
\newcommand\propref[1]{Proposition~\ref{#1}}
\newcommand\secref[1]{Section~\ref{#1}}
\newcommand\appref[1]{Appendix~\ref{#1}}

\newcommand\lemref[1]{Lemma~\ref{#1}}

\begin{document}
\baselineskip=20pt
\hsize=340pt
\vsize=490pt
\title{
       $\hbar$-dependent KP hierarchy
}
\author{Kanehisa Takasaki\\
Graduate School of Human and Environmental Studies,\\
Kyoto University,\\
Yoshida, Sakyo, Kyoto, 606-8501, Japan\\
%
e-mail: takasaki@math.h.kyoto-u.ac.jp\\
\\
Takashi Takebe\\
Faculty of Mathematics,\\
National Research University -- Higher School of Economics,\\
Vavilova Street, 7, Moscow, 117312, Russia\\
e-mail: ttakebe@hse.ru}
\date{}
\maketitle
\begin{abstract}
This is a summary of a recursive construction of solutions of the
 $\hbar$-dependent KP hierarchy. 
%
%
We give recursion relations for the coefficients $X_n$ 
of an $\hbar$-expansion of the operator 
$X = X_0 + \hbar X_1 + \hbar^2X_2 + \cdots$ 
for which the dressing operator $W$ is expressed 
in the exponential form $W = \exp(X/\hbar)$.  
The wave function $\Psi$ associated with $W$ turns out 
to have the WKB form $\Psi = \exp(S/\hbar)$, 
and the coefficients $S_n$ of the $\hbar$-expansion 
$S = S_0 + \hbar S_1 + \hbar^2 S_2 + \cdots$, 
too, are determined by a set of recursion relations. 
This WKB form is used to show that the associated 
tau function has an $\hbar$-expansion of the form 
$\log\tau = \hbar^{-2}F_0 + \hbar^{-1}F_1 + F_2 + \cdots$. 

\end{abstract}

%

\setcounter{section}{-1}
\section{Introduction}
\label{sec:intro}

\commentout{
The KP hierarchy can be completely solved by 
several methods. The most classical methods are 
based on Grassmann manifolds \cite{sat-sat:82}, \cite{seg-wil:85},
fermions and vertex operators \cite{djkm} and 
factorisation of microdifferential operators \cite{mul:84}.
Unfortunately, those methods are not very suited for 
a ``quasi-classical'' ($\hbar$-dependent, where $\hbar$ 
is the Planck constant) formulation \cite{tak-tak:95} 
of the KP hierarchy.  
}

The $\hbar$-dependent formulation of the KP hierarchy 
was introduced to study the dispersionless KP hierarchy 
\cite{kod-gib}, \cite{kri:91}, \cite{tak-tak:91} as a classical limit 
(i.e., the lowest order of the $\hbar$-expansion) 
of the KP hierarchy.   This point of view turned out 
to be very useful for understanding various features 
of the dispersionless KP hierarchy. 
In this paper, we return to the $\hbar$-dependent 
KP hierarchy itself, and consider all orders 
of the $\hbar$-expansion.  

We first address the issue of solving a Riemann-Hilbert problem for the
pair $(L,M)$ of Lax and Orlov-Schulman operators \cite{orl-sch}. This is
a kind of ``quantisation'' of a Riemann-Hilbert problem that solves the
dispersionless KP hierarchy \cite{tak-tak:91}. 
In this paper, we
settle this issue by an $\hbar$-expansion of the dressing operator $W$,
which is assumed to have the exponential form $W = \exp(X/\hbar)$ with
an operator $X$ of negative order.  Roughly speaking, the coefficients
$X_n$, $n = 0,1,2,\ldots$, of the $\hbar$-expansion of $X$ are shown to
be determined recursively from the lowest order term $X_0$ (in other
words, from a solution of the dispersionless KP hierarchy). 

We next convert this result to the language of 
the wave function $\Psi$. 
Namely, given the dressing operator 
in the exponential form $W = \exp (X/\hbar)$, 
we show that the associated wave function has the WKB form 
$\Psi=\exp(S/\hbar)$ with a phase function $S$ expanded 
into nonnegative powers of $\hbar$. 
\commentout{
This is genuinely 
a problem of calculus of microdifferential operators 
rather than that of the KP hierarchy.  A simplest example 
such as $X=x(\hbar\der)^{-1}$ demonstrates that this problem 
is by no means trivial. }%
Borrowing an idea from Aoki's 
``exponential calculus'' of microdifferential operators \cite{aok:86}, 
we show that dressing operators of the form $W = \exp(X/\hbar)$ 
and wave functions of the form $\Psi = \exp(S/\hbar)$ 
are determined from each other by a set of recursion relations 
for the coefficients of their $\hbar$-expansion. 
Consequently, 
the wave function of the solution of the aforementioned 
Riemann-Hilbert problem, too, are recursively determined 
by the $\hbar$-expansion.  

Having the $\hbar$-expansion of the wave function, 
we can readily derive an $\hbar$-expansion of the tau function 
as stated in our previous work \cite{tak-tak:95}.  

\commentout{
This paper is organised as follows.  
Section 1 is a review of the $\hbar$-dependent formulation 
of the KP hierarchy.  Relevant Riemann-Hilbert problems 
are also reviewed here.  
Section 2 presents the recursive solution of 
the Riemann-Hilbert problem.  

Section 3 deals with the $\hbar$-expansion of the wave function. 
Aoki's exponential calculus is also briefly reviewed here.  
Section 4 mentions the $\hbar$-expansion of the tau function. 
}

Details are found in \cite{tak-tak:09} and shall be published
elsewhere.

\section{$\hbar$-dependent KP hierarchy: review}
\label{sec:dkp}

In this section we recall several facts on the KP hierarchy depending on
a formal parameter $\hbar$ in \cite{tak-tak:95}, \S1.7.

The $\hbar$-dependent KP hierarchy is defined by the Lax representation
\begin{equation}
     \hbar \frac{\der L}{\der t_n} = [ B_n, L ], \qquad
    B_n = (L^n)_{\geq 0}, \qquad n=1,2,\ldots,
\label{kph}
\end{equation}
where the {\em Lax operator} $L$ is a microdifferential operator of the
form
\begin{equation}
    L = \hbar \der
      + \sum_{n=1}^\infty u_{n+1}(\hbar,x,t)(\hbar\der)^{-n}, \qquad
    \der = \frac{\der}{\der x},
\label{L}
\end{equation}
and ``$(\quad)_{\geq 0}$'' stands for the projection onto a differential
operator dropping negative powers of $\der$. The coefficients
$u_n(\hbar,x,t)$ of $L$ are formally regular with respect to $\hbar$. 

\commentout{
We need two kinds of ``order'' of microdifferential operators: one is
the ordinary order,
\begin{equation}
    \ord \left( \sum a_{n,m}(x,t) \hbar^n \der^m \right)
    \defeq
    \max \left\{ m 
       \,\left|\, \sum_{n} a_{n,m}(x,t)\hbar^n \neq 0 
         \right.\right\},
\label{def:ord}
\end{equation}
and the other is the $\hbar$-order defined by
}

We introduce the notion of the {\em $\hbar$-order} defined by
\begin{equation*}
    \ordh \left( \sum a_{n,m}(x,t) \hbar^n \der^m \right)
    \defeq
    \max \{ m-n \,|\, a_{n,m}(x,t) \neq 0 \}.
\end{equation*}
In particular, $\ordh\hbar=-1$, $\ordh\der=1$, $\ordh\hbar\der=0$. The
regularity condition which we imposed on the coefficients
$u_n(\hbar,x,t)$ can be restated as $\ordh(L) = 0$.
The {\em principal symbol} 
of a microdifferential operator $A = \sum a_{n,m}(x,t)\hbar^n \der^m$ with
respect to the $\hbar$-order is
$    \symh(A)
    \defeq \sum_{m-n = \ord(A)} a_{n,m}(x,t) \xi^m$.

\commentout{
\begin{rem}
This ``order'' coincides with the order of an microdifferential
operator if we formally replace $\hbar$ with $\der_{t_0}^{-1}$, where
$t_0$ is an extra variable. In fact, naively extending \eqref{kph} to
$n=0$, we can introduce the time variable $t_0$ on which nothing
depends. See also \cite{kas-rou:08}.
\end{rem}
}

As in the usual KP theory, the Lax operator $L$ is expressed by a {\em
dressing operator} $W$:
\begin{equation}
    L = \Ad W (\hbar\der) = W (\hbar\der) W^{-1}
\label{L=Ad(W)d}
\end{equation}
The dressing operator $W$ should have a specific form:
\begin{equation}
    W = \exp( \hbar^{-1} X(\hbar,x,t,\hbar\der) )
        (\hbar\der)^{\alpha(\hbar)/\hbar},
\label{W=exp(X)}
\end{equation}
where $X(\hbar,x,t,\hbar\der) = \sum_{k=1}^\infty \chi_k(\hbar,x, t)
(\hbar\der)^{-k}$ is a $0$-th order operator, $\ordh (X)=0$, and
and $\alpha(\hbar)$ is a constant with respect to $x$ and $t$ with
$\hbar$-order 0, $\ordh \alpha(\hbar) = 0$. 

The {\em wave function} $\Psi(\hbar,x,t;z)$ is defined by
\begin{equation}
    \Psi(\hbar,x,t;z) 
    = W e^{(xz + \zeta(t,z))/\hbar},
\label{def:wave-func}
\end{equation}
where $\zeta(t,z)=\sum_{n=1}^\infty t_n z^n$. It is a solution of
linear equations
$    L \Psi = z \Psi$, 
$    \hbar\frac{\der \Psi}{\der t_n} = B_n \Psi$ 
($n=1,2,\dotsc$)
and has the WKB form 
as we shall show in
\secref{sec:wave-function}. Moreover it is expressed by means of the
{\em tau function} $\tau(\hbar,t)$ as follows:
\begin{equation}
    \Psi(\hbar,x,t;z) =
    \frac{\tau(t+x-\hbar[z^{-1}])}{\tau(t)}
    e^{\hbar^{-1}\zeta(t,z)},
\label{tau/tau}
\end{equation}
where $t+x=(t_1+x,t_2,t_3,\dotsc)$ and
$[z^{-1}]=(1/z,1/2z^2,1/3z^3,\dots)$. We shall study the
$\hbar$-expansion of the tau function in \secref{sec:tau-function}.

The {\em Orlov-Schulman operator} $M$ is defined by 
\begin{equation}
    M = 
    \Ad \left(
          W \exp\left(\hbar^{-1} \zeta(t,\hbar\der) \right)
        \right)x 
    =
    W \left( \sum_{n=1}^\infty n t_n (\hbar \der)^{n-1} + x \right)
    W^{-1}
\label{def:M}
\end{equation}
where $\zeta(t,\hbar\der) = \sum_{n=1}^\infty t_n (\hbar\der)^n$. It is easy
to see that $M$ has a form
\begin{equation}
    M = \sum_{n=1}^\infty n t_n L^{n-1} + x + \alpha(\hbar) L^{-1} +
        \sum_{n=1}^\infty v_n(\hbar,x, t) L^{-n-1},
\label{M}
\end{equation}
and satisfies 
$\ordh(M) = 0$, the canonical commutation relation $[L, M] = \hbar$ and
the same Lax equations as $L$: 
$
    \hbar \frac{\der M}{\der t_n} = [ B_n, M ]$,
$n=1,2,\dotsc$.

The following proposition (Proposition 1.7.11 of \cite{tak-tak:95}) is a 
Riemann-Hilbert type construction of solutions of the $\hbar$-KP
hierarchy. 

\begin{prop}
\label{prop:RH}
(i)
 Let $f(\hbar,x,\hbar\der)$ and $g(\hbar,x,\hbar\der)$ be $0$-th order
 operators ($\ordh f = \ordh g = 0$) and canonically commuting,
 $[f,g]=\hbar$, operators $L$ and $M$ have the form \eqref{L} and
 \eqref{M} respectively and commute canonically, $[L,M]=\hbar$.
 Suppose $f(\hbar,M,L)$ and $g(\hbar,M,L)$ are differential operators:
 $(f(\hbar,M,L))_{<0} = (g(\hbar,M,L))_{<0} = 0$, where $(\quad)_{<0}$
 is the projection to the negative order part: $P_{<0}:=P-P_{\geq 0}$.
 Then $L$ is a solution of the KP hierarchy \eqref{kph} and $M$ is the
 corresponding Orlov-Schulman operator.

(ii)
 Conversely, for any solution $(L,M)$ 
 there exists a pair $(f,g)$ satisfying the conditions in (i).
\end{prop}

The leading term of this system with respect to the $\hbar$-order gives
the {\em dispersionless KP hierarchy}. Namely,
\begin{equation}
    \calL := \symh(L)
    = \xi
    + \sum_{n=1}^\infty u_{0,n+1} \xi^{-n}, \qquad
    (u_{0,n+1} := \symh(u_{n+1}))
\label{calL}
\end{equation}
satisfies the dispersionless Lax type equations
\begin{equation}
    \frac{\der \calL}{\der t_n} = \{ \calB_n, \calL\}, \qquad
    \calB_n = (\calL^n)_{\geq 0}, \qquad n=1,2,\ldots,
\label{dkp}
\end{equation}
where $(\quad)_{\geq 0}$ is the truncation of Laurent series to its
polynomial part and $\{,\}$ is the Poisson bracket defined by
$
    \{a(x,\xi), b(x,\xi)\}
    =
    \frac{\der a}{\der \xi} \frac{\der b}{\der x}
    -
    \frac{\der a}{\der x} \frac{\der b}{\der \xi}$.
%

The dressing operation \eqref{L=Ad(W)d} for $L$ becomes the following
dressing operation for $\calL$:
$    \calL = \exp \bigl( \ad_{\{,\}} X_0 \bigr) \xi$,
where $X_0 = \symh(X)$ and $\ad_{\{,\}} (f) (g):= \{f,g\}$. 

The principal symbol of the Orlov-Schulman operator is
\begin{equation}
    \calM = \sum_{n=1}^\infty nt_n \calL^{n-1} + x 
    + \alpha_0 \calL^{-1}
    + \sum_{n=1}^\infty v_{0,n} \calL^{-n-1}
\label{calM}
\end{equation}
($v_{0,n}=\symh(v_n)$, $\alpha_0=\symh(\alpha)$), which is equal to
$    \calM = 
    \exp \bigl( \ad_{\{,\}} X_0 \bigr) \allowbreak
    \exp \bigl( \ad_{\{,\}} \alpha_0 \log\xi \bigr) \allowbreak
    \exp \bigl( \ad_{\{,\}} \zeta(t,\xi)\bigr) x$,
where $\zeta(t,\xi) = \sum_{n=1}^\infty t_n \xi^n$. The series $\calM$
satisfies the canonical commutation relation $\{\calL,\calM\}=1$, and
the Lax type equations:
%
$
    \frac{\der\calM}{\der t_n} = \{\calB_n, \calM\}$, 
$n=1,2,\ldots$.

The Riemann-Hilbert type construction of the solution is essentially the
same as \propref{prop:RH}. 

\begin{prop}
\label{prop:dRH}
(i)
 Let $f_0(x,\xi)$ and $g_0(x,\xi)$ be functions canonically commuting,
 $\{f_0,g_0\}=1$, functions $\calL$ and $\calM$ have the form
 \eqref{calL} and \eqref{calM} respectively.
 Suppose $f_0(\calM,\calL)$ and $g_0(\calM,\calL)$ do not contain
 negative powers of $\xi$, $(f_0(\calM,\calL))_{<0} =
 (g_0(\calM,\calL))_{<0} = 0$, where $(\quad)_{<0}$ is the projection to
 the negative degree part: $P_{<0}:=P-P_{\geq 0}$.
 Then $\calL$ is a solution of the dispersionless KP hierarchy
 \eqref{dkp} and $\calM$ is the corresponding Orlov-Schulman function.

(ii)
 Conversely, for any solution $(\calL,\calM)$ 
 there exists a pair $(f_0,g_0)$ satisfying the conditions in
 (i). 
\end{prop}

If $f$, $g$, $L$ and $M$ are as in \propref{prop:RH}, then
$f_0=\symh(f)$, $g_0=\symh(g)$, $\calL=\symh(L)$ and
$\calM=\symh(M)$ satisfy the conditions in \propref{prop:dRH}. In other
words, $(f,g)$ and $(L,M)$ are quantisation of the canonical
transformations $(f_0,g_0)$ and $(\calL,\calM)$ respectively. (See, for
example, \cite{sch:85} for quantised canonical transformations.)

\section{Recursive construction of the dressing operator}
\label{sec:recursion}

In this section we prove that the solution of the KP hierarchy
corresponding to the quantised canonical transformation $(f,g)$ is
recursively constructed from its leading term, i.e., the solution of the
dispersionless KP hierarchy corresponding to the Riemann-Hilbert data
$(\symh(f),\symh(g))$.

Given the pair $(f,g)$, we have to construct the dressing operator $W$,
or $X$ and $\alpha$ in \eqref{W=exp(X)}, such that operators
\begin{equation}
 \begin{aligned}
    f(\hbar,M,L) &= 
    \Ad \left(
          W \exp\left(\hbar^{-1} \zeta(t,\hbar\der) \right)
        \right) f(\hbar,x,\hbar\der)
\\
    g(\hbar,M,L) &= 
    \Ad \left(
          W \exp\left(\hbar^{-1} \zeta(t,\hbar\der) \right)
        \right) g(\hbar,x,\hbar\der)
 \end{aligned}
\label{Ad(W)(f,g)}
\end{equation}
are both differential operators (cf.\ \propref{prop:RH}). Let us expand
$X$ and $\alpha$ with respect to the $\hbar$-order as follows:
\begin{equation}
    X(\hbar,x,t,\hbar\der) 
    = \sum_{n=0}^\infty \hbar^n X_n(x,t,\hbar\der),\quad
    X_n(x,t,\hbar\der) =
    \sum_{k=1}^\infty \chi_{n,k}(x,t) (\hbar\der)^{-k},
\end{equation}
$\alpha(\hbar) = \sum_{n=0}^\infty \hbar^n \alpha_n$,
where $\chi_{n,k}$ and $\alpha_n$ do not depend on $\hbar$.

Assume that the solution $(\calL,\calM)$ of the dispersionless KP
hierarchy corresponding to $(\symh(f),\symh(g))$ is given. Namely, 
$(\symh(f)(\calM,\calL), \allowbreak \symh(g)(\calM,\calL))$ do not
contain negative powers of $\xi$: 
\begin{equation}
    \bigl(\symh(f)(\calM,\calL)\bigr)_{<0}=
    \bigl(\symh(g)(\calM,\calL)\bigr)_{<0}=0.
\label{sym(f,g)(L,M)-=0}
\end{equation}
Let $(X_0,\alpha_0)$ be corresponding dressing functions.

\commentout{

In other
words, assume that a symbol $X_0=\sum_{k=1}^\infty \chi_{0,k}(x,t)
\xi^{-k}$ and a constant $\alpha_0$ are given such that
\begin{equation*}
 \begin{aligned}
    \symh(f)(\calM,\calL) &=
    \exp \bigl( \ad_{\{,\}} X_0 \bigr)
    \exp \bigl( \ad_{\{,\}} \alpha_0 \log \xi \bigr)
    \exp \bigl( \ad_{\{,\}} \zeta(t,\xi) \bigr)
    \symh(f)(x,\xi)
\\
    \symh(g)(\calM,\calL) &= 
    \exp \bigl( \ad_{\{,\}} X_0 \bigr)
    \exp \bigl( \ad_{\{,\}} \alpha_0 \log \xi \bigr)
    \exp \bigl( \ad_{\{,\}} \zeta(t,\xi) \bigr)
    \symh(g)(x,\xi)
 \end{aligned}
\end{equation*}
do not contain negative powers of $\xi$:
\begin{equation}
    \bigl(\symh(f)(\calM,\calL)\bigr)_{<0}=
    \bigl(\symh(g)(\calM,\calL)\bigr)_{<0}=0.
\label{sym(f,g)(L,M)-=0}
\end{equation}
(See \propref{prop:dRH}.)
}

We are to construct $X_n$ and $\alpha_n$ recursively, starting from
$X_0$ and $\alpha_0$. The explicit procedure is as follows.

\commentout{

For this purpose expand $f(\hbar,L,M)$ and
$g(\hbar,L,M)$ in \eqref{Ad(W)(f,g)} as
\begin{align}
    P&:= 
    \Ad \left( \exp (\hbar^{-1} X) (\hbar\der)^{\alpha/\hbar} \right)
    f_t
    = \sum_{k=0}^\infty \hbar^k P_k(x,t,\hbar\der),
\label{P}
\\
    Q&:= 
    \Ad \left( \exp (\hbar^{-1} X) (\hbar\der)^{\alpha/\hbar} \right)
    g_t
    = \sum_{k=0}^\infty \hbar^k Q_k(x,t,\hbar\der),
\label{Q}
\end{align}
where 
$
    f_t :=
    \Ad \left( e^{\hbar^{-1} \zeta(t,\hbar\der)} \right) f$, 
$    g_t :=
    \Ad \left( e^{\hbar^{-1} \zeta(t,\hbar\der)} \right) g$,
and 
$\ordh P_i=\ordh Q_i =0$. 
Suppose that we have chosen X_0,\dots,X_{i-1}$ and 
$\alpha_0,\dots,\alpha_{i-1}$ so that $P_0,\dots,P_{i-1}$ and
$Q_0,\dots,Q_{i-1}$ do not contain negative powers of $\der$.


We can construct $X_i$ and $\alpha_i$ recursively such that resulting
$P_i$ and $Q_i$ do not contain negative powers of $\der$ as follows:
}
\begin{itemize}
 \item (Step 0) Assume $X_0,\dots,X_{i-1}$ and
       $\alpha_0,\dots,\alpha_{i-1}$ are given. Set
    $X^{(i-1)} := \sum_{n=0}^{i-1} \hbar^n X_n$,
    $\alpha^{(i-1)} := \sum_{n=0}^{i-1} \hbar^n \alpha_n$.
 \item (Step 1) Set
\begin{align}
    P^{(i-1)} 
    &:=
    \Ad \left( \exp (\hbar^{-1} X^{(i-1)}) 
               (\hbar\der)^{\alpha^{(i-1)}/\hbar} \right)
    f_t,
\label{def:Pi-1}
\\
    Q^{(i-1)}
    &:=
    \Ad \left( \exp (\hbar^{-1} X^{(i-1)}) 
               (\hbar\der)^{\alpha^{(i-1)}/\hbar} \right)
    g_t.
\label{def:Qi-1}
\end{align}
       Expand $P^{(i-1)}$ and $Q^{(i-1)}$ as 
$
    P^{(i-1)} = 
    \sum_{k=0}^\infty\hbar^{k} P^{(i-1)}_k
$,
$    Q^{(i-1)} = 
    \allowbreak
    \sum_{k=0}^\infty \hbar^{k} Q^{(i-1)}_k
$.
        ($\ordh P^{(i-1)}_k=\ordh Q^{(i-1)}_k=0$.)
 \item (Step 2) Put $\calP_0:=\symh(P^{(i-1)}_0)$,
       $\calQ_0:=\symh(Q^{(i-1)}_0)$,
       $\calP^{(i-1)}_i:=\symh(P^{(i-1)}_i)$,
       $\calQ^{(i-1)}_i:=\symh(Q^{(i-1)}_i)$ and define a constant
       $\alpha_i$ and a series $\tilde\calX_i(x,t,\xi)=\sum_{k=1}^\infty
       \tilde\chi_{i,k}(x,t)\xi^{-k}$ by
\begin{equation}
    \alpha_i \log\xi + \tilde\calX_i:=
    \int^\xi\left(
    \dfrac{\der \calQ_0}{\der \xi} \calP^{(i-1)}_i
    -
    \dfrac{\der \calP_0}{\der \xi} \calQ^{(i-1)}_i
    \right)_{\leq -1}
    d\xi.
\label{ai+tildeXi=int}
\end{equation}
       The integral constant of the indefinite integral is fixed so that
       the right hand side agrees with the left hand side.
 \item (Step 3) Define a series $\calX_i(x,t,\xi)=\sum_{k=1}^\infty
       \chi_{i,k}(x,t)\xi^{-k}$ by
\begin{equation}
 \begin{split}
    \calX_i
    &= \tilde \calX'_i
    - \frac{1}{2} \{\symh(X_0),\tilde\calX'_i\}
    + \sum_{p=1}^\infty \frac{B_{2p}}{(2p)!}
      (\ad_{\{,\}} (\symh(X_0)))^{2p} \tilde\calX'_i,
\\
    \tilde\calX'_i &:= 
    \alpha_i \log\xi + \tilde\calX_i(x,\xi)
    - \exp(\ad_{\{,\}} \symh(X_0)) (\alpha_i \log\xi).
 \end{split}
\label{tildeXi->Xi:symbol}
\end{equation}
       Here  $B_{2p}$'s are the Bernoulli numbers. 
 \item (Step 4) The operator $X_i(x,t,\hbar\der)$ is defined as the
       operator with the principal symbol $\calX_i$:
    $X_i=\sum_{k=1}^\infty\chi_{i,k}(x,t)(\hbar\der)^{-k}$.
\end{itemize}

The main theorem is the following:
\begin{thm}
\label{thm:recursion:X}
 Assume that $X_0$ and $\alpha_0$ satisfy \eqref{sym(f,g)(L,M)-=0} and
 construct $X_i$'s and $\alpha_i$'s by the above procedure
 recursively. Then 
 $W=\exp(X/\hbar)(\hbar\der)^{\alpha/\hbar}$ is a dressing operator of
 the $\hbar$-dependent KP hierarchy corresponding to $(f,g)$ by
 \propref{prop:RH}. 
%
\end{thm}

\section{Asymptotics of the wave function}
\label{sec:wave-function}

In this section we prove that the dressing operator of the form
$W(\hbar,x,t, \hbar\der) = \exp(X(\hbar,x,\hbar\der)/\hbar)$,
$\ordh X \leqq 0$, $\ord X \leqq -1$,
gives a wave function of the form
$    \Psi(\hbar,x,t;z) 
    = W e^{(xz + \zeta(t,z))/\hbar} = \exp(S(\hbar,x,t,z)/\hbar)$,
$    S(\hbar,x,t;z)
    = \sum_{n=0}^\infty \hbar^n S_n(x,t;z) + \zeta(t,z)$,
$\zeta(t,z) := \sum_{n=1}^\infty t_n z^n$,
and vice versa. 

Since the time variables $t_n$ do not play any role in this section,
we set them to zero. As the factor $(\hbar\der)^{\alpha/\hbar}$ in
\eqref{W=exp(X)} becomes a constant factor $z^{\alpha/\hbar}$ when it is
applied to $e^{xz/\hbar}$, we also omit it here.

Let $A(\hbar,x,\hbar\der)=\sum_n a_n(\hbar,x) (\hbar\der)^n$ be a
microdifferential operator. The {\em total symbol} of $A$ is a power
series of $\xi$ defined by
$
    \totsym(A)(\hbar,x,\xi):= \sum_n a_n(\hbar,x) \xi^n
$,
or, equivalently defined by the formula
$A e^{xz/\hbar} = \totsym(A)(\hbar,x,z) e^{xz/\hbar}$.
%
\begin{prop}
\label{prop:exp(:X:)=:exp(S):}
 Let $X=X(\hbar,x,\hbar\der)$ be a microdifferential operator such that
 $\ord X=-1$ and $\ordh X = 0$. Then the total symbol of $e^{X/\hbar}$
 has such a form as
$    \totsym(\exp(\hbar^{-1} X(\hbar,x,\hbar\der)))
    = e^{S(\hbar,x,\xi)/\hbar}$,
 where $S(\hbar,x,\xi)$ is a power series of $\xi^{-1}$ without
 non-negative powers of $\xi$ and has an $\hbar$-expansion
$    S(\hbar,x,\xi)=\sum_{n=0}^\infty \hbar^n S_n(x,\xi)$.

 Moreover, the coefficient $S_n$ is determined by $X_0,\dots,X_n$ in the
 $\hbar$-expansion 
 of $X=\sum_{n=0}^\infty \hbar^n X_n$.
\end{prop}

We omit the explicit formula for $S_n$. (See \cite{tak-tak:09}.)

\begin{prop}
\label{prop::exp(S):=exp(:X:)}
 Let $S=\sum_{n=0}^\infty \hbar^n S_n$ be a power series of $\xi^{-1}$
 without non-negative powers of $\xi$. Then there exists a
 microdifferential operator $X(\hbar,x,\hbar\der)$ such that $\ord X
 \leqq -1$, $\ordh X \leqq 0$ and
$
    \totsym(\exp(\hbar^{-1} X(\hbar,x,\hbar\der)))
    = e^{S(\hbar,x,\xi)/\hbar}$.
 Moreover, the coefficient $X_n(x,\xi)$ in the $\hbar$-expansion
 $X=\sum_{n=0}^\infty \hbar^n X_n$ of the total symbol
 $X=X(\hbar,x,\xi)$ is determined by $S_0,\dots,S_n$ in the
 $\hbar$-expansion of $S$.
 
\end{prop}

We omit the explicit formula for $X_n$. (See \cite{tak-tak:09}.)

Combining these propositions with the results in \secref{sec:recursion},
we can, in principle, make a recursion formula for $S_n$
($n=0,1,2,\dots$) of the wave function of the solution of the KP
hierarchy corresponding to the quantised canonical transformation
$(f,g)$ as follows: let $S_0,\dots,S_{i-1}$ be given.
\begin{enumerate}
 \item By \propref{prop::exp(S):=exp(:X:)} we have $X_0,\dots,X_{i-1}$.
 \item We have a recursion formula for $X_i$ by
       \thmref{thm:recursion:X}. 
 \item \propref{prop:exp(:X:)=:exp(S):} gives a formula for $S_i$.
\end{enumerate}
If we take the factor $(\hbar\der)^{\alpha/\hbar}$ into account, this
process becomes a little bit complicated, but essentially the same.

\section{Asymptotics of the tau function}
\label{sec:tau-function}

In this section we derive an $\hbar$-expansion 
    $\log\tau(\hbar,t) = \sum_{n=0}^\infty \hbar^{n-2} F_n(t)$
of the tau function (cf.\ \eqref{tau/tau}). 
Note that we have suppressed the variable $x$, which is
understood to be absorbed in $t_1$.

\commentout{

Let us recall the fundamental relation \cite{djkm} between the wave
function \eqref{def:wave-func} and the tau function again:
\begin{equation}
    \Psi(t;z) =
    \frac{\tau(t-\hbar[z^{-1}])}{\tau(t)}
    e^{\hbar^{-1}\zeta(t,z)},
    \qquad
    [z^{-1}] 
    =
    \left(\frac{1}{z},\frac{1}{2z^2}, \frac{1}{3z^3},\dots \right),
\end{equation}
where $\zeta(t,z)=\sum_{n=1}^\infty t_n z^n$. 
}

The logarithmic derivation of \eqref{tau/tau} gives
\commentout{

This implies that 
\begin{equation}
    \hbar^{-1} \hat S(t;z)
    =
    \left(
    e^{-\hbar D(z)} - 1
    \right) \log\tau(t)
\label{tau->w}
\end{equation}
where $\hat S(t;z) = S(t;z)- \zeta(t,z)$ and 
$D(z) =\sum_{j=1}^\infty \frac{z^{-j}}{j} \frac{\der}{\der t_j}$. 
Differentiating this with respect to $z$, we have
\begin{equation}
 \begin{split}
    \hbar^{-1} \frac{\der}{\der z} \hat S(t;z)
    =&
    - \hbar D'(z) e^{-\hbar D(z)} \log\tau(t)
\\
    =&
    - \hbar D'(z) (\hbar^{-1} \hat S(t;z) + \log\tau(t)),
 \end{split}
\label{(tau->w)'}
\end{equation}
where $\hat S(t;z) = S(t;z)- \zeta(t,z)$, 
$D(z) =\sum_{j=1}^\infty \frac{z^{-j}}{j} \frac{\der}{\der t_j}$
and $D'(z):=\frac{\der}{\der z} D(z) = - \sum_{j=1}^\infty z^{-j-1}
\frac{\der}{\der t_j}$. Hence
}
\begin{equation}
    -\hbar D'(z) \log\tau(t)
    =
    \hbar^{-1} \left(\frac{\der}{\der z} + \hbar D'(z)\right)
    \hat S(t;z),
\label{w->tau}
\end{equation}
where $\hat S(t;z) = S(t;z)- \zeta(t,z)$ 
and $D'(z):=
- \sum_{j=1}^\infty z^{-j-1} \frac{\der}{\der t_j}$. 

\commentout{

Thus 
we obtain 
\begin{equation}
    \hbar\frac{\der}{\der t_n} \log \tau(t)
    =
    \hbar^{-1} \Res_{z=\infty} z^n 
    \left(\frac{\der}{\der z} + \hbar D'(z)\right) \hat S(t;z)\, dz,
\label{dtau/dtn}
\end{equation}
for $n=1,2,\dotsc$. 
}

By substituting the $\hbar$-expansions
$
    \log\tau(t) = \sum_{n=0}^\infty \hbar^{n-2} F_n(t)$,
$
    \hat S(t;z) = \sum_{n=0}^\infty \hbar^n S_n(t;z)$,
\commentout{

into \eqref{w->tau}, we have
\begin{equation}
    \sum_{j=1}^\infty \sum_{n=0}^\infty z^{-j-1} \hbar^{n-1}
    \frac{\der F_n}{\der t_j}
    =
    \sum_{n=0}^\infty \left(
    \hbar^{n-1} \frac{\der S_n}{\der z}
    - 
    \sum_{j=1}^\infty z^{-j-1} \hbar^n \frac{\der S_n}{\der t_j}
    \right).
\label{S->F}
\end{equation}
Let us expand $S_n(t;z)$ into a power series of $z^{-1}$:
}
and expanding $S_n(t;z)$ as
$S_n(t;z) = - \sum_{k=1}^\infty \frac{z^{-k}}{k} v_{n,k}$,
we have the equations 
\begin{equation}
    \frac{\der F_n}{\der t_j}
    =
    v_{n,j} 
    + \sum_{\substack{k+l=j\\ k\geq 1, l\geq 1}}
      \frac{1}{l}\, \frac{\der v_{n-1,l}}{\der t_k}
    \qquad (v_{-1,j}=0).
\label{dFn/dtj}
\end{equation}
This system determines $F_n$ up to integration constants. 

Combining this with the results in \secref{sec:recursion},
\secref{sec:wave-function}, if $F_0$, which determines a solution of the
dispersionless KP hierarchy, and the corresponding quantised canonical
transformation $(f,g)$ are given, we can construct $F_n$ recursively and
consequently the tau function of a solution of the $\hbar$-KP
hierarchy.

\bigskip
%
\paragraph*{\em Acknowledgments}

The authors are grateful to Professor Akihiro Tsuchiya 
for drawing our attention to this subject.
This work is partly supported by Grants-in-Aid for Scientific Research
No.\ 19540179 and No.\ 22540186 from the Japan Society for the Promotion
of Science.  TT is partly supported by the grant of the National
Research University -- Higher School of Economics, Russia, for the
Individual Research Project 10-01-0043 (2010).

This is a contribution to the Proceedings of the ``International
Workshop on Classical and Quantum Integrable Systems 2011'' (January
24-27, 2011 Protvino, Russia). TT thanks the organisers of the workshop
for hospitality.

\end{document}